\documentclass[creativecommons,noncommercial]{eptcs}
\providecommand{\event}{Dice-Fopara 2017}

\usepackage[T1]{fontenc}
\usepackage[utf8]{inputenc}
\usepackage{amssymb,amsmath}
\usepackage{amsthm}
\usepackage{stmaryrd}
\usepackage{xspace}
\usepackage{tikz}
\usepackage{enumitem}

\usepackage{breakurl}

\definecolor{grey}{gray}{0.5}
\usepackage{hyperref}
\hypersetup{
  colorlinks=true,
  linkcolor=grey,
  citecolor=black,
  filecolor=magenta,
  urlcolor=blue,
  pdfstartview=FitV,  
  pdfpagelayout=TwoPageRight,  
  pdfpagemode=UseOutlines,
  pdftitle={Computability in the Lattice of Equivalence Relations},
  pdfauthor={J.-Y.~Moyen \& J.~G.~Simonsen},
  pdfsubject={},
  pdfkeywords={ICC} {Decidability} {Lattice}
}

\usepackage{amssymb}
\usepackage{newunicodechar}

\newunicodechar{…}{\ldots}
\newunicodechar{ }{~} 

\newunicodechar{α}{\alpha}
\newunicodechar{β}{\beta}
\newunicodechar{γ}{\gamma} \newunicodechar{Γ}{\Gamma}
\newunicodechar{δ}{\delta} \newunicodechar{Δ}{\Delta}
\newunicodechar{ε}{\epsilon}
\newunicodechar{ζ}{\zeta}
\newunicodechar{η}{\eta}
\newunicodechar{θ}{\theta} \newunicodechar{Θ}{\Theta}
\newunicodechar{ι}{\iota}
\newunicodechar{κ}{\kappa}
\newunicodechar{λ}{\lambda} \newunicodechar{Λ}{\Lambda}
\newunicodechar{μ}{\mu}
\newunicodechar{ν}{\nu}
\newunicodechar{ξ}{\xi} \newunicodechar{Ξ}{\Xi}
\newunicodechar{π}{\pi} \newunicodechar{Π}{\Pi}
\newunicodechar{ρ}{\rho}
\newunicodechar{σ}{\sigma} \newunicodechar{Σ}{\Sigma}
\newunicodechar{τ}{\tau}
\newunicodechar{υ}{\upsilon} \newunicodechar{Υ}{\Upsilon}
\newunicodechar{φ}{\varphi} \newunicodechar{Φ}{\Phi}
\newunicodechar{χ}{\chi}
\newunicodechar{ψ}{\psi} \newunicodechar{Ψ}{\Psi}
\newunicodechar{ω}{\omega} \newunicodechar{Ω}{\Omega}

\newunicodechar{≤}{\leq}
\newunicodechar{≥}{\geq}
\newunicodechar{≠}{\neq}
\newunicodechar{±}{\pm}
\newunicodechar{¬}{\neg}
\newunicodechar{×}{\times}
\newunicodechar{≃}{\simeq}
\newunicodechar{≲}{\lesssim}
\newunicodechar{≳}{\grtsim}

\newunicodechar{₁}{_1} \newunicodechar{¹}{^1}
\newunicodechar{₂}{_2} \newunicodechar{²}{^2}
\newunicodechar{₃}{_3} \newunicodechar{³}{^3}
\newunicodechar{₄}{_4} \newunicodechar{⁴}{^4}
\newunicodechar{₅}{_5} \newunicodechar{⁵}{^5}
\newunicodechar{₆}{_6} \newunicodechar{⁶}{^6}
\newunicodechar{₇}{_7} \newunicodechar{⁷}{^7}
\newunicodechar{₈}{_8} \newunicodechar{⁸}{^8}
\newunicodechar{₉}{_9} \newunicodechar{⁹}{^9}
\newunicodechar{₀}{_0} \newunicodechar{⁰}{^0}

\newtheorem{theorem}{Theorem}[section]
\newtheorem{lemma}[theorem]{Lemma}
\newtheorem{proposition}[theorem]{Proposition}
\newtheorem{corollary}[theorem]{Corollary}

\newcommand{\binrep}[1]{\ensuremath{\mathtt{#1}}}
\newcommand{\upset}[1]{\ensuremath{\uparrow \{ #1 \}}\xspace}
\newcommand{\size}[1]{\left\vert #1 \right\vert}

\newcommand{\logicor}{\mathbin{\vert\vert}}
\newcommand{\logicand}{\mathbin{\&\&}}

\newcommand{\eg}{\emph{e.g.}\xspace}

\newcommand{\bN}{\ensuremath{\mathbb{N}}\xspace}

\let\union\bigcup
\let\inter\bigcap

\newcommand{\setof}[2]{\left\{\,#1\ :\ #2\,\right\}}

\newcommand{\sigclass}[2]{\ensuremath{Σ^{#1}_{#2}}}
\newcommand{\piclass}[2]{\ensuremath{Π^{#1}_{#2}}}
\newcommand{\deltclass}[2]{\ensuremath{Δ^{#1}_{#2}}}
\newcommand{\sigarith}[1]{\sigclass{0}{#1}}
\newcommand{\piarith}[1]{\piclass{0}{#1}}
\newcommand{\deltarith}[1]{\deltclass{0}{#1}}

\newcommand{\compclass}[1]{\textsc{#1}\xspace}
\newcommand{\logspace}{\compclass{LogSpace}}
\newcommand{\ptime}{\compclass{Ptime}}
\newcommand{\pspace}{\compclass{Pspace}}
\newcommand{\exptime}{\compclass{ExpTime}}

\newcommand{\aut}{\compclass{Aut}}

\newcommand{\Riceeq}{\ensuremath{\mathfrak{R}}\xspace}
\newcommand{\Equ}[2][]{\ensuremath{\textrm{Equ}(#2)_{#1}}\xspace}
\newcommand{\thelat}[1][]{\Equ[#1]{\bN}}

\newcommand{\equspec}[1]{\thelat[#1]}
\newcommand{\equre}{\equspec{\mathrm{r.e.}}}

\newcommand{\equpr}{\equspec{\mathrm{p.r.}}}
\newcommand{\equexp}{\equspec{\exptime}}
\newcommand{\equpspace}{\equspec{\pspace}}
\newcommand{\equp}{\equspec{\ptime}}
\newcommand{\equl}{\equspec{\logspace}}

\newcommand{\equaut}{\equspec{\aut}}

\newcommand{\equ}[1]{\ensuremath{\mathcal{#1}}\xspace}
\newcommand{\equone}{\equ{E}}
\newcommand{\equtwo}{\equ{F}}
\newcommand{\equthree}{\equ{G}}

\newcommand{\toeven}{\to_{\mathrm{even}}}
\newcommand{\toodd}{\to_{\mathrm{odd}}}
\newcommand{\approxeven}{\approx_{\mathrm{even}}}
\newcommand{\approxodd}{\approx_{\mathrm{odd}}}

\newcommand{\lubname}{join\xspace}
\newcommand{\glbname}{meet\xspace}
\newcommand{\lub}{\ensuremath{\vee}\xspace}
\newcommand{\LUB}{\ensuremath{\bigvee}\xspace}
\newcommand{\glb}{\ensuremath{\wedge}\xspace}
\newcommand{\GLB}{\ensuremath{\bigwedge}\xspace}

\newcommand{\parti}[1]{\ensuremath{\mathcal{#1}}\xspace}

\title{Computability in the Lattice of Equivalence Relations}

\author{Jean-Yves Moyen\footnote{Supported by the Marie
    Sk\l{}odowska--Curie action ``Walgo'', program H2020-MSCA-IF-2014,
    number 655222} \qquad \qquad Jakob Grue
  Simonsen\footnote{Partially supported by the Danish Council for
    Independent Research \emph{Sapere Aude} grant
    ``Complexity via Logic and Algebra'' (COLA).}\\
  \institute{Department of Computer Science, University of Copenhagen
    (DIKU)\\
    Njalsgade 128-132, 2300 Copenhagen S, Denmark}
  \email{\hspace{-1cm} Jean-Yves.Moyen@univ-paris13.fr \qquad \qquad
    simonsen@diku.dk}
}

\providecommand{\authorrunning}{J.-Y.~Moyen \& J.~G.~Simonsen}
\providecommand{\titlerunning}{Computability in the Lattice of
  Equivalences}

\date{\today}

\begin{document}

\maketitle

\begin{abstract}
  We investigate computability in the lattice of equivalence relations
  on the natural numbers. We mostly investigate whether the subsets of
  appropriately defined subrecursive equivalence relations --for
  example the set of all polynomial-time decidable equivalence
  relations-- form sublattices of the lattice.
\end{abstract}

\section{Introduction}
\subsection{Motivation}
As there are uncountably many properties of programs (because each set
of programs is a property), most properties must be undecidable. A
celebrated negative result shows that there are even further barriers
to decidability: Rice's Theorem~\cite{Rice} states that every
non-trivial, extensional property on programs is undecidable. An
extensional property is one that depends solely on the (input-output)
\emph{function} computed by the program.

While Rice's theorem is six decades old, it still spurs new research
trying to find the boundary of decidable progam properties, for
example Asperti's work on complexity cliques~\cite{Asperti:2008}. Such
properties can be studied fruitfully by viewing results such as Rice's
and Asperti's as assertions about \emph{equivalence relations};
indeed, any program property is an equivalence relation that has at
most two classes: the class of programs having the property and the
class of programs not having it. In this view, Rice's Theorem studies
the \emph{extensional equivalence}, \Riceeq, where two programs are
equivalent if and only if they compute the same function. Rice's
theorem thus states that \emph{no non-trivial equivalence class or
  union of equivalence classes in \Riceeq is decidable}.

Rather than studying individual equivalences such as \Riceeq, it is
interesting to look at the set of all equivalences between programs
and at subsets of equivalences sharing certain properties. The set of
all equivalence relations between programs has a very natural complete
lattice structure. With this order, taking the union of some classes
yields a larger equivalence relation and Rice's Theorem says that any
class of any non-trivial equivalence relation in the principal filter
at \Riceeq is undecidable.

The fact that Rice's Theorem is so neatly expressed in the lattice
language hints that we can gain knowledge by studying it. One way is
to look at other equivalences that the extensional
one~\cite{Asperti:2008, MS-Dice}. Another is to look at the lattice
structure itself. Since there are uncountably many equivalences, the
lattice is hard to study and we want to find a good way to approximate
it in a manageable way. For this reason, we study subsets of
equivalences and how they interact with the lattice structure.

The lattice-theoretic properties of $\Equ{S}$ are well-understood;
basic facts can be gleaned from the papers
\cite{Ore42,RivalStanford:algpart} and the textbooks \cite[\S
8-9]{Birkhoff:lattice} \cite[Sec.\ IV.4]{Gratzer:genlattice}.  As the
lattice structure of $\Equ{S}$ is isomorphic to $\Equ{T}$ for any sets
$S$ and $T$ of identical cardinality, we may without loss of
generality consider the set of equivalences over the natural numbers,
\thelat. That is, we work modulo an unspecified encoding of
programs.

Because Rice's Theorem is all about computability, we are interested
here in the subsets of equivalences that are defined in term of
computability (\emph{e.g.} the set of decidable equivalences) or
complexity (\emph{e.g.} the set of equivalences decidable in
polynomial time). Rice's Theorem can thus be expressed as saying that
the intersection of the principal filter at $\Riceeq$ and the set of
decidable equivalences is reduced to the trivial equivalence with one
class containing everything.

Finding a subset of equivalences that is manageable but still retain
most of the order-theoretic properties of the whole lattice provides a
way to approximate \thelat (as $\mathbb{Q}$ approximates $\mathbb{R}$)
and allows for a more focused and systematic study.

\subsection{Equivalence relations and Lattice}
We consider the set \thelat of all equivalences over natural
numbers. If \equone is an equivalence, we write $m \equone n$ to say
that $m$ and $n$ are equivalent. \thelat is ordered by
$\equone \leq \equone'$ iff $m \equone n \Rightarrow m \equone'
n$. Note that this is exactly the subset ordering $\subseteq$ on
$\bN \times \bN$. We can easily see that $\equone ≤ \equone'$ iff
each class of $\equone'$ is the union of one or more classes of
\equone.

$(\thelat, ≤)$ is a lattice with the following operations:
\begin{itemize}[nosep]
\item The \glbname (greatest lower bound) of \equone and \equtwo is
  $\equthree = \equone \glb \equtwo$ such that $m \equthree n$ iff $m
  \equone n$ and $m \equtwo n$. In this case, the classes of \equthree
  are exactly the (non-empty) intersections of one class of \equone
  and one class of \equtwo.
\item The \lubname (lowest upper bound) of \equone and \equtwo is
  $\equthree = \equone \lub \equtwo$ such that $m \equthree n$ iff
  there exists a finite sequence $a_1, …, a_k$ such that $m
  \equone a_1 \equtwo a_2 \equone … \equtwo n$.
\end{itemize}

Note that the \lubname is much harder to express (and compute) than
the \glbname. This will be reflected in the results. Indeed closure
for \glbname usually boils down to closure under intersection, but
closure for \lubname is far from closure under union.

Standard results for \thelat were laid out by Ore~\cite{Ore42}. It is
known that \thelat is, among others:
\begin{itemize}[nosep]
\item bounded with minimal element $\bot$ being the equivalence where
  each class is a singleton (no two different elements are equivalent)
  and the maximal element $\top$ being the equivalence with a single
  class containing every elements (any two elements are equivalents);
\item complete, that is closed under arbitrary \lubname and \glbname
  (and not only under finite ones);
\item atomistic, each element is the \lubname of atoms, where atoms
  are successors of $\bot$, that is equivalences with one class
  containing two elements and the other are singletons;
\item relatively complemented (hence complemented), for each \equone,
  there exists \equtwo such that $\equone \lub \equtwo = \top$ and
  $\equone \glb \equtwo = \bot$; non $\top$ or $\bot$ equivalences
  have infinitely many complements, most have uncountably many.
\end{itemize}

An equivalence relation \equone where exactly one class is not a
singleton is called \emph{singular}.

\begin{lemma}[Complements to singular equivalence relations]
  \label{lem:compl-sing}
  Let \equone be a singular equivalence relation with non-singleton class
  $E$. \equtwo is one of its complement iff each class of \equtwo
  contains exactly one element of $E$.
\end{lemma}

To avoid confusion with meet and join, we note $\logicand$ and
$\logicor$ the logical conjunction and disjunction. We note
$\binrep{n}$ the binary representation of $n$.

\subsection{Computability and complexity}
We refer to standard textbooks convering computability and complexity
theory (\eg, \cite{JonesCC}).

Unless otherwise stated, all Turing Machines are multi-tape machines
with two designated read-only input tapes, one write-only output tape,
and any number of work tapes. Machines are deterministic unless
otherwise stated.

Equivalences relations are then classified in the natural way with
respect to the Turing Machine deciding them. That is, for example,
``polynomial time relations'' or ``recursively enumerable relations''
are defined in the straightforward way.

\subsection{Results}
The paper considers three broad categories: automatic, subrecursive,
and arithmetical equivalence relations and how they interact with the
basic properties of the entire lattice of equivalence relations,
namely \glbname, \lubname and complements. The results are summarised
here.

\begin{center}
  \begin{tabular}{|r|cc|cc|c|c|}
    \hline
    & \multicolumn{2}{c|}{Finite}
    & \multicolumn{2}{c|}{Arithmetical}
    & Arbitrary & \\
    & $\glb$ & $\lub$ & $\glb$
                &  $\lub$ &  $\glb$/$\lub$ & complements\\
    \hline
    Automatic & Yes & Yes & No & Yes & No/Yes & N/A\\
    Subrecursive & Yes & No$^{\dagger}$ & ? & No$^{\dagger}$
    & No & $≥$ \pspace\\
    $\sigarith{k}$ & Yes & Yes & No & Yes & No & No\\
    $\piarith{k}$ & Yes & No & Yes & No & No & ?\\
    $\deltarith{k}$ & Yes & No & No & No & No & Yes\\
    \multicolumn{7}{|l|}{\footnotesize $^{\dagger}$: for
    \logspace or larger classes.}\\
    \hline
  \end{tabular}
\end{center}

Note that a ``Yes'' entry does not imply any results concerning actual
computation of the operation --it merely implies that the subset of
equivalence relations is closed under the operation considered-- but
the actual computation involved in the operation may require resources
beyond those implied by the subset. If we consider a subset $S$ of
equivalences and an operation on the equivalences, there are four
possibilities:
\begin{enumerate}[nosep]
\item $S$ is not closed under this operation, that is there exist
  elements in $S$ such that, when applying the operation to them we
  can obtain an element not in $S$. \emph{This is the canonical
    meaning of ``No'' in the above table.}
\item \label{foo} $S$ is closed under the operation (i.e., there is a
  ``Yes'' in the above table), and either
  \begin{enumerate}[nosep]
  \item \label{foolmeonce} the operation is computable within the
    resource requirements defining the class, e.g. polytime computable
    for the set of polytime decidable equivalence relations.
  \item \label{foolmetwice} the operation is computable, but not
    within the resource requirements defining the class.
  \item \label{foolmethrice} the operation is not computable.
  \end{enumerate}
\end{enumerate}

\section{Automatic equivalence relations}
There are several ``natural'' ways to define equivalence relations
decidable by finite automata. The one we consider here uses a single,
two-way input tape with two inputs separated by a special symbol;
another natural variation would have the two inputs on two separate
tapes, or a single tape but two heads. Unlike Turing machines, the
class of sets decidable by finite automata properly increases with the
number of input tapes and tape heads; it is therefore quite
conceivable that simple variations of the class we consider would have
different properties.

We consider two-way finite automata that have only a single tape where
both inputs are encoded. We believe the most natural form to be
$a \square b$ where $a,b \in \{0,1\}^+$ and $\square$ is a separator
symbol (``blank''). We denote $\equaut$ the set of \emph{automatic}
equivalence relations, that is such that the language
$\{\mathtt{m} \square \mathtt{n} : m\parti{E}n\}$ is accepted by a
two-way non-deterministic finite automaton. By standard results, this
language is thus regular. That is, 2-way Non-Deterministic Automata
recognise exactly the same languages as 1-way Deterministic Automata.

Once such an automaton is fixed, for any integer $m$ we note $q_m$ the
state reached after reading $\binrep{m} \square$.

\begin{proposition}\label{prop:automata_silly}
  Let $\parti{E} \in \equaut$. It has finitely many equivalence
  classes.
\end{proposition}

The proof uses an argument taken from the Myhill-Nerode
Theorem~\cite{nerode1958}.

\begin{proof}[Sketch of proof]
  The automaton must accept $\binrep{m} \square \binrep{m}$ by
  reflexivity of equivalences. Hence, by determinism, if $q_n = q_m$,
  the automaton must also accept $\binrep{n} \square \binrep{m}$ and
  $n \equone m$. Thus, there can be at most as many classes as states.
\end{proof}

Note that Proposition \ref{prop:automata_silly} implies that some very
simple equivalence relations are not elements of $\equaut$; for
example, $\bot \notin \equaut$ (that is, $\{\binrep{n} \square
\binrep{n}\}$ is not a regular language, a well-known fact).

Since the number of classes in any equivalence of $\equaut$ is finite
but unbounded, given such an equivalence, it is always possible to
find a class with several elements and create a new equivalence by
taking one element out of this class and making a new singleton
class. This new equivalence is smaller than the initial one and is
still single-tape automatic. Hence, $\equaut$ is neither bounded, nor
complete.

Moreover, let $\equone$ and $\equtwo$ be two equivalences, each with
finitely many classes $E_1, …, E_m$ and $F_1, …, F_n$. Since the
classes of $\equone \glb \equtwo$ are exactly the non-empty
$E_i \inter F_j$, there are at most $n × m$ such classes, that is, a
finite number, thus it cannot be $\bot$ and $\equtwo$ cannot be a
complement to $\equone$.

Hence, no complement of a single-tape automatic equivalence is
single-tape automatic.

\begin{theorem}
  $\equaut$ is a lattice.
\end{theorem}

\begin{proof}[Sketch of proof]
  The language for the meet is the intersection of both languages and
  regular languages are closed under intersection.

  In order to build an automaton recognising the join, we first start
  by finding representatives of each class of both equivalences :
  $m₁, …, m_e$ and $n₁, …, n_f$; this is doable by standard search in
  minimal DFA. Next we precompute which $m_i$ and $n_j$ are related in
  $\equone \lub \equtwo$; this is doable because the finite number of
  classes implies a bounded length for the chain when unfolding the
  join definition.

  Lastly, we can build an NFA with $ε$-transitions that first check
  which unique $m_i$ is such that $m \equone m_i$ ($ε$-transitions to $e$
  copies of the first automaton); and then check if $n$ is such that
  $n \equtwo n_j$ for one of the $n_j$ related to $m_i$ (correct
  number of $ε$-transitions to copies of the second automata).
\end{proof}

Thus, not only is $\equaut$ a lattice, but both the meets and the
joins are effectively computable. Computing these, however, cannot be
performed by a DFA (this is case~\ref{foolmetwice} of the discussion
after the results Table).

\begin{proposition}
  Let $k ≥ 1$. The \glbname of a $\sigarith{k}$ set of automatic
  equivalence relations is not necessarily automatic.
\end{proposition}

\begin{proof}
  Let, for each $i ≥ 1$, $\equone_i$ be the automatic equivalence
  relation containing the two classes $\{i\}$ and
  $\mathbb{N} \setminus \{i\}$. Then, $\GLB \equone_i = \bot$, which
  is not automatic.
\end{proof}

\begin{lemma}\label{lem:equaut_upward_closed}
  $\equaut$ is upward closed.
\end{lemma}

\begin{proof}
  Let $\equone \in \equaut$ and $\equone ≤ \equtwo$. By definition,
  the classes of $\equtwo$ are unions of classes of $\equone$ and
  because there are only finitely many classes in $\equone$, these are
  \emph{finite} unions. Thus, building a DFA for $\equtwo$ is possible.
\end{proof}

\begin{proposition}
  Let $A \subseteq \equaut$ be non-empty. Then, $\LUB A \in \equaut$.
\end{proposition}

\begin{proof}
  Because $\equone ≤ \LUB A$ for any $\equone \in A$.
\end{proof}

As mentioned at the beginning of the Section, the expressive power of
automata varies with the number of tapes or heads. Hence,
small variations of the model can drastically change the property
of the corresponding subset of equivalences. For example, two-tape
automata can test whether both tape contain the same word and thus
decide $\bot$.

Even the choice of representation of inputs may affect the expressive
power. For example, instead of sequencing the inputs
($\binrep{m}\square\binrep{n}$), it is possible to interleave them
($m_1n_1m_2n_2\ldots m_in_i \ldots \square n_k$ if $n$ is longer than
$m$). This representation allows an automaton to decide $\bot$ (note that the Myhill-Nerode
argument does not work in this case).

Moreover, when working with multi-tape automata, the question of
synchronicity arises: should the tapes be ``at the same position'' on
each tape, or can the heads move independently?

Thus, we have only studied one particular class of automata with one
particular way of representing equivalences. $\equaut$ hence enjoys
properties that other classes of ``automatic'' relation may or may not
have. While the class of single-tape, single-head DFAs studied above
is the simplest kind of finite automaton --hence most apt to study
first-- it remains to perform a systematic study of equivalence
relations decidable by other, more particular, kinds of automata.

\section{Subrecursive equivalence relations}\label{sec:subrecursive}
We now treat classes of equivalence relations decidable within bounds
on their resources. The definition of such classes are simply the
standard definitions of computational complexity theory using Turing
machines with two input tapes. The sets of equivalence relations on
$\bN$ consisting of \logspace-, \ptime-, \pspace-, \exptime-,
primitive recursive, … decidable equivalence relations are denoted by
$\equl$, $\equp$, $\equpspace$, $\equexp$, $\equpr$, … The notation is
extended in straightforward ways and we collectively call these
``subrecursive'' sets of equivalences.

\begin{lemma}
  The subrecursive sets of equivalences are closed under finite
  \glbname.
\end{lemma}

\begin{proof}
  $m (\equone \glb \equtwo) n$ iff $m \equone n \logicand m \equtwo
  n$, and the subrecursive classes are closed under $\logicand$.
\end{proof}

\begin{proposition}\label{prop:logspace_not_upper}
  There exist two equivalence relations, decidable in logarithmic
  space, whose \lubname is undecidable.

  Thus, the subrecursive sets of equivalences are not closed under
  join and are not sublattices.
\end{proposition}

\begin{proof}[Sketch of proof]
  Let $M$ be a deterministic Turing Machine. We define the
  \emph{clocked one-step} relation by
  $(\binrep{n, c}) \to (\binrep{n', c'})$ iff $n'=n+1$ and
  $\binrep{c'}$ is the (representation of the) configuration resulting
  from executing one step $M$ from $\binrep{c}$; or if $\binrep{c}$ is
  a final configuration and $(\binrep{n', c'}) = (0,0)$. $\to$ is
  decidable in logarithmic space by standard techniques.

  We now define $\toeven$ (resp. $\toodd$) as the restriction of $\to$
  for even (resp. odd) $n$; and $\approxeven$ (resp. $\approxodd$) as
  the reflexive, transitive, symmetric closure of $\toeven$
  (resp. $\toodd$). Because of the parity restriction, it is not
  possible to have
  $(\binrep{n₀, c₀}) \toeven (\binrep{n₁, c₁}) \toeven (\binrep{n₂,
    c₂})$ and thus $\approxeven$ is also decidable in logarithmic
  space.

  Let $\approx = \left(\approxeven \lub \approxodd\right)$. The
  computation starting at $c$ terminates iff
  $(\binrep{n, c}) \approx (0, 0)$. Hence, $\approx$ is not decidable.
\end{proof}

\begin{theorem}\label{the:S_is_closed}
  Let $A$ be a (deterministic) subrecursive set of equivalences larger
  than $\equpspace$ (included) and $\equone \in A$. There is at least
  one complement to $\equone$ in $A$.
\end{theorem}

\begin{proof}[Sketch of proof]
  Let $\equtwo$ be the singular equivalence whose non-singleton class,
  $F$, contains exactly the least element of each class of
  $\equone$. It is a complement to $\equone$ by
  Lemma~\ref{lem:compl-sing}. Because the deterministic subrecursive
  sets are closed under complement, $\overline{\equone} \in A$.

  To decide if $m \equtwo n$, proceed as follows: (i) if
  $\binrep{m} = \binrep{n}$, accept; (ii) if there exists $k < m$ with
  $\binrep{k} \equone \binrep{m}$, reject; (iii) if there exists
  $k' < n$ with $\binrep{k'} \equone \binrep{n}$, reject; (iv) if not
  rejected yet, accept.

  Step (ii) is done by checking if for all $k < m$,
  $k \overline{\equone} m$ and requires $O(\size{\binrep{m}})$ working
  space for writing $k$, hence at least linear space. It also loops
  over $m = 2^{\size{\binrep{m}}}$ values and thus needs exponential
  time.
\end{proof}

It is tempting to conjecture that there are equivalences in $\equp$
with no complement in $\equp$. However if that were the case, we would
have $\equp \neq \equpspace$, and by using a polynomial-time pairing
function $\mathbb{N}^2 \longrightarrow \mathbb{N}$, existence of
$\parti{P} \in \equpspace \setminus \equp$ entails existence of a set
in $\pspace \setminus \ptime$, and hence $\pspace \neq \ptime$; hence,
the conjecture will be exceedingly hard to prove. On the other hand,
It is not clear that equality of $\equp$ and $\equpspace$ would entail
$\mathrm{PSPACE} = \mathrm{P}$.

\section{Arithmetical equivalence relations}
We now investigate closure properties of sets of arithmetical
equivalence relations. Closure under finite \glbname is immediate as
the arithmetical sets are also closed under finite intersections.

\begin{lemma}\label{lem:sigma_closed_by_lub}
  For every $k≥1$, the set of equivalence relations in
  $\sigarith{k}$ is closed under finite \lubname.

  The \lubname of any two $\piarith{k}$ equivalence relations is in
  $\sigarith{k+1}$.
\end{lemma}

\begin{proof}
  Let $x$ and $y$ be two integers. By definition,
  $x(\equone \lub \equtwo) y$ iff there exists $a_1, …, a_n$ such that
  $x \equone a_1 \equtwo … \equone a_n \equtwo y$. That is
  $\exists n, a_1, …, a_n . x \equone a_1 \logicand a_1 \equtwo a_2
  \logicand … \logicand a_n \equtwo y$, which is $\sigarith{k}$ (resp.
  $\sigarith{k+1}$) if $\equone$ and $\equtwo$ are both $\sigarith{k}$
  (resp. $\piarith{k}$).
\end{proof}

\begin{proposition}\label{prop:delta_not_closed_by_lub}
  For any $k ≥ 1$, the set of equivalence relations in $\deltarith{k}$
  is not closed under finite \lubname.
\end{proposition}

\begin{proof}[Idea of proof]
  The proof is essentially the same as for
  Proposition~\ref{prop:logspace_not_upper}. The ``hard'' problem
  cannot stay the classical halting problem and must be replaced by
  the halting problem for machines with oracle in $\deltarith{k}$.
\end{proof}

\begin{proposition}\label{prop:pi_not_closed_by_lub}
  For every $n \geq 1$, the set of equivalence relations in
  $\piarith{n}$ is not closed under finite \lubname.
\end{proposition}

\begin{proof}
  Because $\deltarith{k} = \sigarith{k} \cap \piarith{k}$ is not, but
  $\sigarith{k}$ is.
\end{proof}

\begin{theorem}
  For every $k ≥ 1$, there are equivalence relations in $\sigarith{k}$
  none of whose complements are in $\sigarith{k}$.
\end{theorem}

\begin{proof}[Sketch of proof]
  Let $E$ be a $\sigarith{k}$ set whose complement is not
  $\sigarith{k}$, and $\equone$ be the singular equivalence with
  non-singleton class $E$. Let $\equtwo$ be a complement to $\equone$,
  each of its class intersects $E$ in exactly one point, hence
  $x \in \overline{E}$ iff
  $\exists e . e \neq x \logicand e \in E \logicand e \equtwo x$, and
  $\equtwo$ cannot be $\sigarith{k}$.
\end{proof}

\begin{theorem}
  Let $k ≥ 0$. Every equivalence relation in $\deltarith{k}$ has at
  least one complement in $\deltarith{k}$.
\end{theorem}

The proof is essentially the same as for Theorem~\ref{the:S_is_closed}.

\section{Infinite \glbname and \lubname}
Recall that the upper set \upset{n} contains $n$ and all the elements
greater than $n$. We say that an equivalence is \emph{small} if it
has finitely many classes. Note that small singular equivalence
relations have finitely many singleton classes, whose largest element
is $N$ and the non-singleton class is thus the union of a finite set
and the upper set \upset{N+1}.

\begin{proposition}\label{prop:get_any_singular}
  Let $I \subseteq \mathbb{N}$, there is a set of small singular
  equivalences whose \glbname is singular with non-singleton class
  $I$. Thus, none of the subrecursive or arithmetical sets of
  equivalences are closed under arbitrary \glbname.
\end{proposition}

\begin{proof}[Sketch of proof]
  Let $I \subseteq \bN$. Let $(f_i)_{i \in \bN}$ be a strictly
  increasing sequence and let $F_i = I \inter \left[0;f_i\right[$ and
  $F^+_i = (F_i \union \upset{f_i})$. Note that $I = \inter_i
  F^+_i$. Let $\equ{F}_i$ be the small singular equivalence with
  non-singleton class $F^+_i$.

  $\GLB \equtwo_i$ is the singular equivalence whose non-singleton
  class is $I$.
\end{proof}

\begin{proposition}
  Let $I \subseteq \bN$. There exists a set of atoms whose \lubname is
  singular with non-singleton class $I$.

  Thus, none of the arithmetical sets of equivalences are closed under
  arbitrary \lubname.
\end{proposition}

\begin{proof}
  Because the lattice is atomistic. Atoms (singular equivalences whose
  non-singleton set has only two elements) are decidable in zero space
  and linear time (by encoding the two elements in the states of the
  machine).
\end{proof}

The previous results assume that the set of equivalence relations may
be chosen arbitrarily. Especially, it is based on a set $I$ of
arbitrary difficulty. This means that the results boil down to the
fact that there are uncountably many such $I$, while there are only
countably many arithmetical relations.

We now look at what happens if the set of relations itself is to have
some bound. Specifically, we are concerned with the \glbname and
\lubname of a $\sigarith{k}$ set of equivalences. $\sigarith{k}$ (and
especially, recursively enumerable) is a sensible bound: It means that
one can enumerate each of the equivalences (with a proper oracle, or
none for recursively enumerable) until all the needed ones have been
found. On the other hand, a $\piarith{k}$ set of equivalences would
mean that one can enumerate (with oracle) the complement to this set,
which is less practical.

\begin{proposition}
  For all $k ≥ 1$, the \lubname of a $\sigarith{k}$ set of
  $\sigarith{k}$ equivalences is a $\sigarith{k}$ equivalence.
\end{proposition}

\begin{proof}
  Let $\setof{Φ_i}{i \in I}$ be a $\sigarith{k}$ set of $\sigarith{k}$
  equivalences, that is $I$ is a $\sigarith{k}$ set of integers and
  each $Φ_i$ is a $\sigarith{k}$ equivalence. Let
  $\equone = \LUB Φ_i$. By definition, we have $x \equone y$ if and
  only if
  \[
  \exists n, a_1, …, a_n, x_1, …, x_n . \quad a_1 \in I \logicand …
  \logicand a_n \in I \quad \logicand x Φ_{a_1} x_1 \logicand x_1
  Φ_{a_2} x_2 \logicand … \logicand x_n Φ_{a_n} y
  \]
\end{proof}

Note that because $i ≤ j$ implies $\sigarith{i} \subset \sigarith{j}$,
we immediately have that the \lubname of a $\sigarith{i}$ set of
$\sigarith{j}$ equivalences is $\sigarith{j}$. In particular, the
\lubname of a recursively enumerable set of $\sigarith{j}$
equivalences is a $\sigarith{j}$ equivalence relation.

\begin{proposition}
  For all positive integers $n$, there exists a $\sigarith{n}$ set of
  $\sigarith{n}$ equivalence relations whose \glbname is not
  $\sigarith{n}$.
\end{proposition}

\begin{proof}[Sketch of proof]
  We consider an encoding of formulae into numbers and let $I$ be the
  set of encoding of $\sigarith{k}$ formulae, it is decidable if the
  encoding allows to count quantifiers.  Let
  $A_k = \setof{i \in I}{Φ_i(k)}$ be the set of $\sigarith{n}$
  equivalence accepting $k$, it is $\sigarith{n}$. Let $\equ{A}_k$ be
  the singular equivalence whose non-singleton set is $A_k$, it is a
  $\sigarith{k}$ equivalence.

  Now, $\GLB \equ{A}_k$ is singular with non-singleton class the
  (encoding of) $\sigarith{k}$ tautologies. This is a
  $\piarith{k+1}$-complete set.
\end{proof}

\begin{proposition}
  For all $k$, the \glbname of a $\sigarith{k}$ set of $\piarith{k}$
  equivalence relations is $\piarith{k}$.
\end{proposition}

\begin{proof}
  Let $\setof{Φ_i}{i \in I}$ be a $\sigarith{k}$ set of $\piarith{k}$
  equivalences and $\equone = \GLB Φ_i$. By definition, $x \equone y$
  iff $\forall i, i \in I \Rightarrow x Φ_i y$, which is equivalent to
  $\forall i, i \notin I \logicor x Φ_i y$, a $\piarith{k}$ formula.
\end{proof}

Here also, by inclusion of the hierarchy, $i ≤ j$ implies that any
$\sigarith{i}$ set of $\piarith{j}$ equivalence relations has a $\piarith{j}$
\glbname. Especially, the set of $\piarith{k}$ equivalence relations is closed
under recursively enumerable \glbname.

\begin{proposition}\label{prop:equl_not_glb_re}
  There is an r.e.\ set of elements of $\equl$ whose \glbname is not
  decidable.
\end{proposition}

\begin{proof}[Sketch of proof]
  Let $\equone_n$ be the singular equivalence whose non-singleton set
  is the (representation of) Turing Machines that \emph{do not} halt
  in $n$ or less step. It is a \logspace equivalence by clever
  encoding of TMs (see, \emph{e.g.}, \cite[Ch. 3]{Pap94}), the main
  trick here being that we only need to simulate a \emph{fixed}
  number of steps ($n$) and this requires a \emph{fixed} amount of
  extra space (plus logarithmic overhead for the simulation). However,
  $\GLB \equone_n$ is the singular equivalence whose non-singleton
  class is the Turing Machines who never halt and is thus not decidable.
\end{proof}

\begin{corollary}
  The subrecursive sets of equivalences are not closed under
  arithmetical meets.

  The sets of $\piarith{k}$ equivalences are not closed under
  arithmetical meets.
\end{corollary}

For the second point, the proof must be adapted using the halting
problem for the correct class of oracle machines.

\section{Conclusion and Future Works}
Out of the three classes of equivalences that we have considered, the
automatic ones seem too few to be of interest; the subrecursive ones
do not form a sublattice, making them bad candidates for studying the
lattice structure; but the arithmetical ones seem more
interesting. Notably, the $\sigarith{k}$ sets of equivalences do keep
the lattice structure and especially $\equre$ is worth more efforts.

In addition to being the smallest $\sigarith{k}$ set of equivalences,
$\equre$ is also linked back to the starting point via the
Rice-Shapiro's Theorem~\cite{Myhill-Shepherdson:rice-shapiro,
  Shapiro:rice-shapiro}: while $\Riceeq \notin \equre$, we already
know something about equivalences in the intersection of $\equre$ and
the principal filter at $\Riceeq$.

Since there are equivalences in $\equre$ with no complements in it, we
do not get all the basic lattice property (the sublattice is not
complemented). Thus, we may still want to find other sublattices
defined by other criteria.

\smallskip

Other questions raised by this work include a more systematic study of
the subsets decided by various kind of automata; and trying to build,
for each equivalence in $\equp$, a complement in $\equp$.

\bibliographystyle{eptcs}
\bibliography{biblio}

\end{document}